\documentclass[letterpaper, 10 pt, conference]{ieeeconf}  

\IEEEoverridecommandlockouts
\newtheorem{property}{Property}

\newtheorem{lemma}{Lemma}

\usepackage[ruled,vlined]{algorithm2e}
\usepackage{subcaption}
\usepackage{cite}
\usepackage{amsmath,amssymb,amsfonts}
\usepackage{graphicx}
\usepackage{textcomp}
\usepackage{xcolor}

\title{\LARGE \bf
Minimum Clustering of Matrices Based on Phase Alignment
}

\author{
Honghao Wu,
Kemi Ding and 
Li Qiu, \textit{Fellow,~IEEE}
\thanks{*This work was supported in part by the Natural Science Foundation of Guangdong Province No.2024A1515011630}
\thanks{Honghao Wu, Kemi Ding are with the School of Automation and Intelligent Manufacturing (AIM), Southern University of Science and Technology, Shenzhen 518055, China. Email: {\tt\small 12332671 @mail.sustech.edu.cn; dingkm@sustech.edu.cn.}}
\thanks{Li Qiu is with the School of Science and Engineering, The Chinese University of Hong Kong, Shenzhen, Shenzhen 518172, China. Email: 
{\tt\small qiuli@cuhk.edu.cn}.}
}

\begin{document}

\maketitle
\thispagestyle{empty}
\pagestyle{empty}

\begin{abstract}
Coordinating multi-agent systems requires balancing synchronization performance and controller implementation costs. To this end, we classify agents by their intrinsic properties, enabling each group to be controlled by a uniform controller and thus reducing the number of unique controller types required.
Existing centralized control methods, despite their capability to achieve high synchronization performance with fewer types of controllers, suffer from critical drawbacks such as limited scalability and vulnerability to single points of failure. On the other hand, distributed control strategies, where controllers are typically agent-dependent, result in the type of required controllers increasing proportionally with the size of the system. 

This paper introduces a novel phase-alignment-based framework to minimize the type of controllers by strategically clustering agents with aligned synchronization behaviors. Leveraging the intrinsic phase properties of complex matrices, we formulate a constrained clustering problem and propose a hierarchical optimization method combining recursive exact searches for small-scale systems and scalable stochastic approximations for large-scale networks. This work bridges theoretical phase analysis with practical control synthesis, offering a cost-effective solution for large-scale multi-agent systems. The theoretical results applied for the analysis of a 50-agent network illustrate the effectiveness of the proposed algorithms.
\end{abstract}

\section{INTRODUCTION}
The coordination of multi-agent systems has become a pivotal area of research due to its extensive applications in fields such as robotics, autonomous vehicles, and distributed sensor networks, see textbooks and research monographs \cite{b23,b24,b25,b26}. In these systems, multiple agents interact within a shared environment, necessitating the development of control strategies that enable coherent and efficient collective behavior. Extending beyond traditional engineering applications, the principles of multi-agent coordination have been applied to social networks, where individuals or entities are modeled as agents within a complex system. In such contexts, understanding and influencing the behavior of these agents is crucial for tasks like information dissemination, opinion dynamics, and collective decision-making.

Historically, centralized control approaches were employed to manage multi-agent systems \cite{b27}. In this paradigm, a single controller possesses global knowledge of the system and dictates the actions of all agents. While this method ensures coordinated behavior, it suffers from significant drawbacks, including scalability issues and vulnerability to single points of failure. Studies have shown that centralized control systems, despite their effectiveness in enforcing coordinated strategies, struggle with increasing system complexity, leading to computational inefficiencies and bottlenecks \cite{b1}. Additionally, designing such centralized controllers becomes increasingly challenging as the number of agents grows due to the combinatorial explosion of possible system states. Research also indicates that centralized architectures are highly susceptible to single-point failures, making them less robust in dynamic environments \cite{b2}. For example, centralized models used in multi-agent deep reinforcement learning have been found to experience scalability constraints due to computational and communication overhead \cite{b3}. Furthermore, energy-efficient multi-agent systems find decentralized or hybrid control methods to be more scalable and resilient compared to purely centralized architectures \cite{b4}.

To address these limitations, distributed control strategies have been developed \cite{b24,b25}, wherein each agent independently determines its actions based on local information and limited communication with neighboring agents. This approach enhances the system’s robustness and scalability, as it eliminates reliance on a central authority and allows for parallel processing. Studies have shown that distributed control enables asymptotic consensus in multi-agent systems, improving overall stability and coordination \cite{b5}. However, a notable disadvantage of distributed control is the requirement for a large number of different types of agent-dependent controllers, which can lead to increased complexity in coordination and potential communication overhead \cite{b5}.

Despite significant advances in the synchronization and coordination of multi-agent systems, one critical yet less-explored aspect is the number of controller types employed within these systems. Existing distributed or decentralized methods typically assume that each agent or subgroup of agents requires an individual controller, thus increasing both implementation complexity and cost, especially for large-scale networks. Recent studies have shown that reducing the number of controller types without sacrificing synchronization performance is a key challenge \cite{b6}. Although containment control methods partially mitigate this challenge by reducing the number of controller types through leader-follower schemes \cite{b7}, determining the minimal number of controller types necessary for synchronization while maintaining desired performance remains an open and challenging problem \cite{b8}.

Motivated by recent advances in phase theory \cite{b9}, which provide a fresh perspective in addressing synchronization problems by capturing the inherent phase relationships among agents \cite{b10,b11,b12}, this paper proposes to exploit these intrinsic phase properties to minimize the required number of controller types. Specifically, by leveraging the phase characteristics inherent in agent clusters, we aim to strategically identify and group agents whose synchronization behaviors are naturally aligned, thus reducing redundant control inputs and simplifying the overall network control structure.

The proposed framework achieves an effective balance between synchronization performance and practical implementation cost. By considering both synchronization conditions derived from phase theory and constraints on the number of controller types, our approach offers a systematic methodology to achieve synchronization using the fewest possible controllers. This not only reduces complexity and cost but also improves system scalability and robustness in practical deployment.

The rest of this paper is organized as follows. In Section \uppercase\expandafter{\romannumeral 2}, preliminaries on graph theory and phase theory are introduced. In Section \uppercase\expandafter{\romannumeral 3}, we briefly review the multi-agent synchronization problem within the framework of phase theory and formulate a minimum clustering problem with phase constraint. We design dual algorithms in Section \uppercase\expandafter{\romannumeral 4}: one identifies exact global optima, and the other approximates them computationally. Simulation results are presented in Section \uppercase\expandafter{\romannumeral 5} and Section \uppercase\expandafter{\romannumeral 6} concludes this paper. 
\section{PRELIMINARIES}
\subsection{Phase Theory}
A complex number $v = ae^{j\theta}$ is fundamentally composed of two components: magnitude $a$ and phase $\theta$. For a complex matrix $A\in \mathbb{C}^{m\times m}$, its magnitude is conventionally characterized by singular values $\sigma(A)$, while the question of how to define the matrix phase has long remained open. Recently we have established a rigorous and well-structured definition for matrix phases \cite{b13,b14}.

The numerical range of a matrix $A$ is defined as
\begin{equation*}
    W(A)=\{x^*Ax:x\in\mathbb{C}^m \,\text{with}\, \|x\|=1\}.
\end{equation*}

The numerical range $W(A)$ of matrix $A$, depicted in the complex plane as shown in Fig. \ref{fig2}(a), forms a convex set. A matrix is called sectorial if its numerical range resides entirely within a half-plane without origin. For such matrices, the numerical range can be bounded by two tangent rays, as illustrated in Fig. \ref{fig2}(b). Let us denote
\begin{equation}
    \phi(A)=[\phi_1(A)\; \,\phi_2(A)\,\cdots \,\phi_m(A)],
\end{equation}
where $\phi_1(A)\geq\phi_2(A)\geq\dots\geq\phi_m(A)$, as the phases of sectorial matrix $A$.
The phase $\phi_i(A)$ lies within the closed interval $[\overline{\phi}(A),\underline{\phi}(A)]$, where $\overline{\phi}(A)$ and $\underline{\phi}(A)$ represent the minimum and maximum phases of the matrix, respectively, with all intermediate phases contained within this interval. It has been proven that for a sectorial matrix $A^{m\times m}$ \cite{b15}, there exists a unique sectorial decomposition:
\begin{equation}
    A=T^*DT,
\end{equation}
where $T$ is nonsingular and
\begin{equation}
    D=\text{diag}\{e^{i\phi_1(A)},\,e^{i\phi_2(A)},\dots,\,e^{i\phi_m(A)}\}.
\end{equation}
\begin{figure}[htbp]
\centerline{\includegraphics[width=0.43\textwidth]{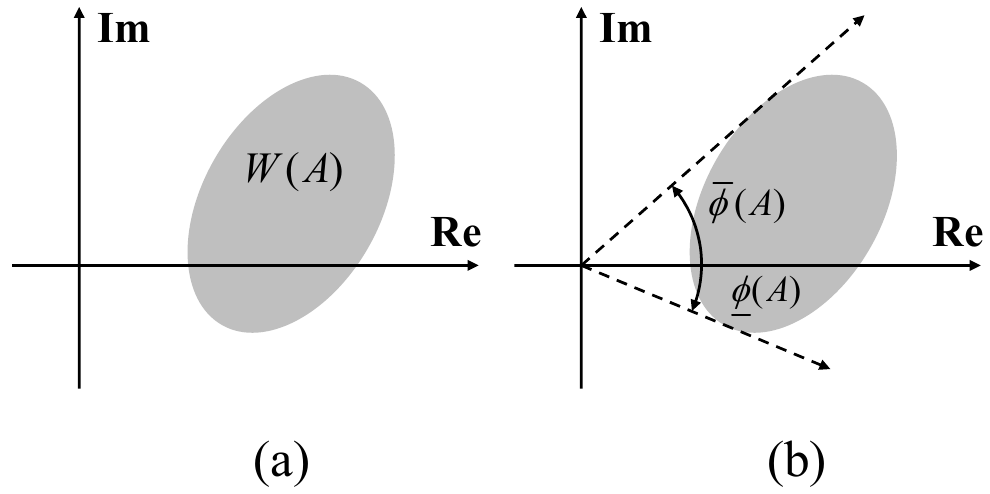}}
\caption{Numerical range of sectorial matrix.}
\label{fig2}
\end{figure}
Furthermore, when the numerical range $W(A)$ lies in the right half-plane and contains the origin (indicating that $A$ is singular), we encounter two special cases.
When the origin is a sharp point of the numerical range as shown in Fig. \ref{fign2}(a), we say matrix $A$ is quasi-sectorial. When the numerical range is tangent to the origin as depicted in Fig. \ref{fign2}(b), the matrix is termed semi-sectorial. For these two special cases, $|\phi(A)|=\text{rank}(A)$. The detailed computational approach can be found in \cite{b14}.
\begin{figure}[htbp]
\centerline{\includegraphics[width=0.43\textwidth]{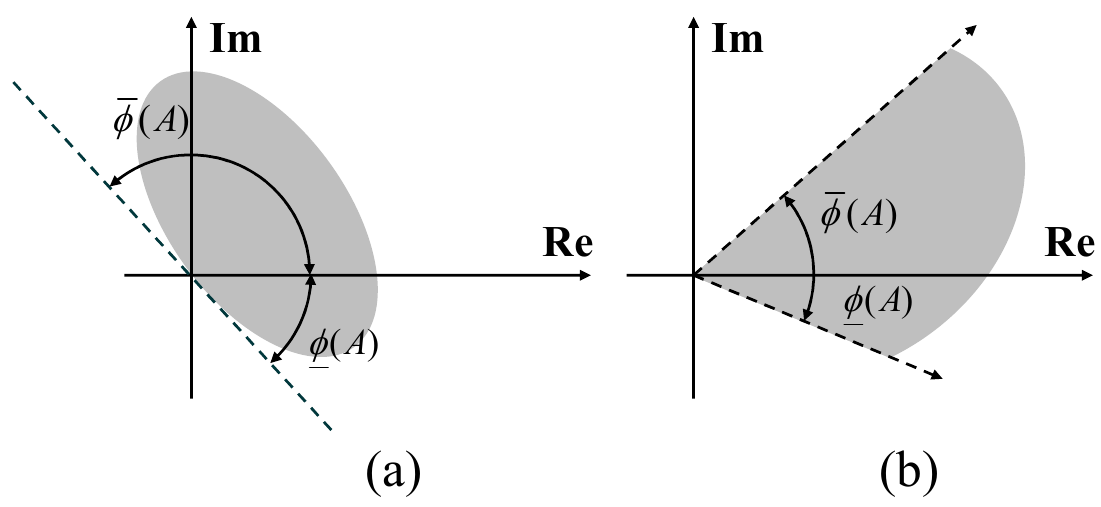}}
\caption{Origin is on the boundary of the numerical range.}
\label{fign2}
\end{figure}

However, most common situations, such as transfer matrices in control systems, feature matrices in machine learning, and social networks, are unlikely non-sectorial, nor are they even square matrices. For these cases, we can obtain their phase characteristics using the following methods.

Given an angle $\alpha\in(-\pi/2,\,\pi/2)$, for a complex matrix $A\in\mathbb{C}^{m\times n},\,m\leq n$, if there exists a matrix $K\in\mathbb{C}^{m\times n}$ satisfying
\begin{equation}
    \phi(AK) \in [-\alpha,\alpha]^r,
    \label{eq: single_alignable}
\end{equation}
with $\text{rank}(AK)=\text{rank}(A)=r$, we define $A$ as an $\alpha$-alignable matrix. 
If $A$ is nonsingular, we can use its inverse $A^{-1}$ as $K$, thereby obtaining $AA^{-1}=I$. It is evident that the phases of $I$ are zero, thus making $A$ alignable for any arbitrary $\alpha$. 
If $A$ is a singular matrix with $\text{rank}(A)=r$, we can take its pseudoinverse as $K$, obtaining
\begin{equation*}
    AA^{\dagger} = U
    \begin{pmatrix}
    I_{r} & 0 \\
    0 & 0
    \end{pmatrix}U^*,
\end{equation*}
where $(\cdot)^{\dagger}$ denotes the Moore-Penrose generalized inverse and $U$ is unitary. The phases of this product are identically zero.

For a set of complex matrices of the same dimensions $\mathcal{C} = \left\{ A_i \in \mathbb{C}^{m\times n} : i = 1, 2, \dots, k, \, m \leq n \right\}$, if there exists a matrix $K^{n\times m}$ such that
\begin{equation}
    \phi(A_iK)\in[-\alpha,\alpha]^{r_i},\quad i=1,\dots,k,
    \label{eq: multi_alignable}
\end{equation}
with $\text{rank}(A_iK)=\text{rank}(A_i)=r_i$. We refer to these matrices as simultaneously $\alpha$-alignable,  as illustrated in Fig. \ref{fig3}.
By leveraging the geometric properties of the numerical range in the complex plane, the uniform control matrix $K$ can be obtained by solving a group of LMIs \cite{b28}:
\begin{equation}
    \begin{aligned}
    \text{Re}(A_iK) &\geq KK^H,\\
    \text{tan}(\alpha)\text{Re}(A_iK) \geq \text{Im}(A_iK)
     &\geq -\text{tan}(\alpha)\text{Re}(A_iK),
    \end{aligned}
    \label{eq: LMIs}
\end{equation}
where $i=1,\dots,k$, $\text{Re}\{\cdot\}$ and $\text{Im}\{\cdot\}$ denote the real and imaginary parts of a complex matrix respectively.
\begin{figure}[htbp]
\centerline{\includegraphics[width=0.4\textwidth,page=2]{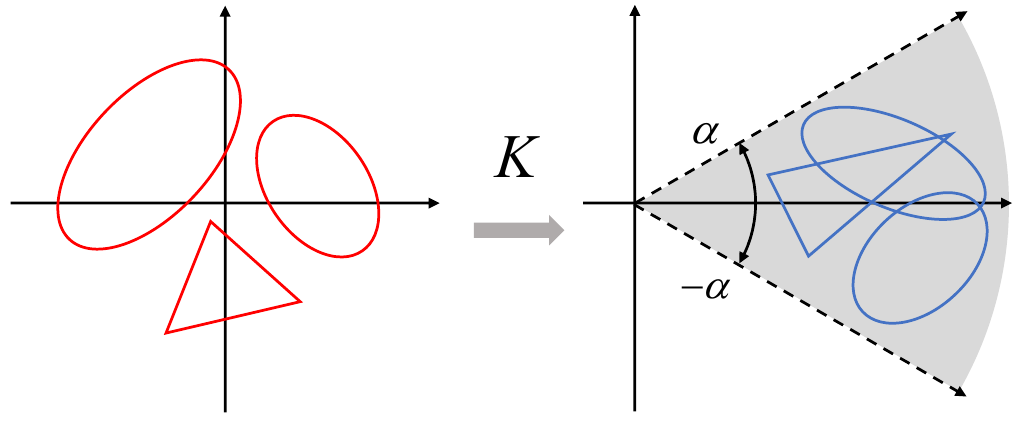}}
\caption{Numerical ranges of matrices before (left) and after (right) alignment.}
\label{fig3}
\end{figure}

\begin{property}[Downward closedness]
\label{downwardclosed}
    If matrix set $\mathcal{C}$ is $\alpha$-alignable, then any $\mathcal{C}' \subseteq \mathcal{C}$ is also $\alpha$-alignable.
\end{property}

This property ensures that larger or equal subsets being alignable implies all their smaller subsets remain alignable as well. Moreover, it also serves as the foundation for proving Lemma \ref{lem:swap-alpha-align}.

Furthermore, for a matrix set $\mathcal{C}$, let
\begin{equation*}
    \alpha(\mathcal{C})=\{\alpha\in[0,\pi/2):\,\mathcal{C}\;\text{is simultaneously $\alpha$-alignable} \}.
\end{equation*}
We had defined the \emph{diversity} of matrix set $\mathcal{C}$ in \cite{b28}:
\begin{equation}
    \mathrm{div}(\mathcal{C}) =
\begin{cases}
    \inf \alpha(\mathcal{C}), & \text{if } \,\alpha(\mathcal{C})\, \text{ is nonempty},\\
    \dfrac{\pi}{2}, & \text{otherwise}.
    \end{cases}
\end{equation}
\subsection{Graph Theory}
A \textit{directed graph} (digraph) is denoted as $\mathbb{G} = (\mathcal{V}, \mathcal{E})$, where $\mathcal{V} = \{1, ..., n\}$ represents the vertex set and $\mathcal{E} \subseteq \mathcal{V} \times \mathcal{V}$ denotes ordered edge pairs. For undirected graphs, edges $(i, j)$ and $(j, i)$ are considered identical.

A graph is \textit{connected} if there exists a path between every pair of vertices. In digraphs, we distinguish between \textit{strong connectivity} (bidirectional reachability for all vertex pairs) and \textit{weak connectivity} (connectivity in the underlying undirected graph). Any graph can be decomposed into \textit{connected components}—maximal subgraphs where every vertex pair is connected. For digraphs, these are either \textit{strongly connected components} (maximal subgraphs with mutual reachability) or \textit{weakly connected components} (components in the underlying undirected graph).

The graph structure can be encoded in an \textit{adjacency matrix} $W \in \mathbb{R}^{n \times n}$, where entries $w_{ij}$ indicate edge existence from vertex $i$ to $j$. For \textit{weighted graphs}, $w_{ij} \in \mathbb{R}^+$ represents connection strength. When weights encode similarity measures, we obtain a \textit{similarity matrix} $S$ satisfying $s_{ij} = s_{ji} \geq 0$ with zero diagonal elements.

The graph \textit{Laplacian matrix} plays a crucial role in spectral analysis:
\begin{equation}
    L = D - W,
\end{equation}
where $D$ is the degree matrix with diagonal entries $d_{ii} = \sum_{j=1}^n w_{ij}$. This matrix exhibits key properties including positive semi-definiteness, $L\mathbf{1} = \mathbf{0}$, and a zero eigenvalue multiplicity equal to the number of connected components.

The Laplacian matrix \( L \) plays a crucial role in phase theory applications. 
As shown in Theorem 9.2 of \cite{b14}, for the Laplacian matrix of a strongly connected graph, it has a positive left eigenvector $v$ related to the zero eigenvalue. The essential phase of $L$ can be obtained by:
\begin{equation}
\label{essp}
    \phi_{\text{ess}}(L) = \overline{\phi} (D^{-1} L D) = \overline{\phi} (V L),
\end{equation}
where $V=\text{diag}\{v\}$ and $D=V^{-1/2}$.
\section{PROBLEM FORMULATION}
Consider a situation in a network of $k$ agents, where the parameters of each agent can be characterized by a matrix of particular dimension $C_i\in \mathbb{C}^{m\times n}$, $i = 1,\dots,k$, $m\leq n$. We aim to design multiple controllers to regulate their behavior.
For agent $i$, if there exists a matrix $K$ such that $K$ and the feature matrix $C_i$ of the agent satisfy \eqref{eq: single_alignable}, then agent $i$ is said to be controllable via the controller $K$.

To control all agents, the ideal scenario would be, as in \eqref{eq: multi_alignable}, to design a uniform controller $K$ that can regulate all agents simultaneously \cite{b10}. However, in typical multi-agent systems, it is generally not possible for all agents to be simultaneously alignable.

The most straightforward solution is to design agent-dependent controllers for all agents. However, when the control system scales to a large number of agents, this approach can become highly resource-consuming.

Therefore, the corresponding problem arises: how to partition these matrices into $p$ clusters, such that all matrices within each cluster can be simultaneously aligned under the given angle $\alpha$ by the controller $K_j,\,j\in1,2,\dots,p$, as shown in Fig. \ref{fig4}. This implies that the diversity of all clusters is less than or equal to $\alpha$. When $p$ reaches its minimum value, the number of controller types required to regulate all agents is also minimized.
\begin{figure}[htbp]
\centerline{\includegraphics[width=0.5\textwidth,page = 1]{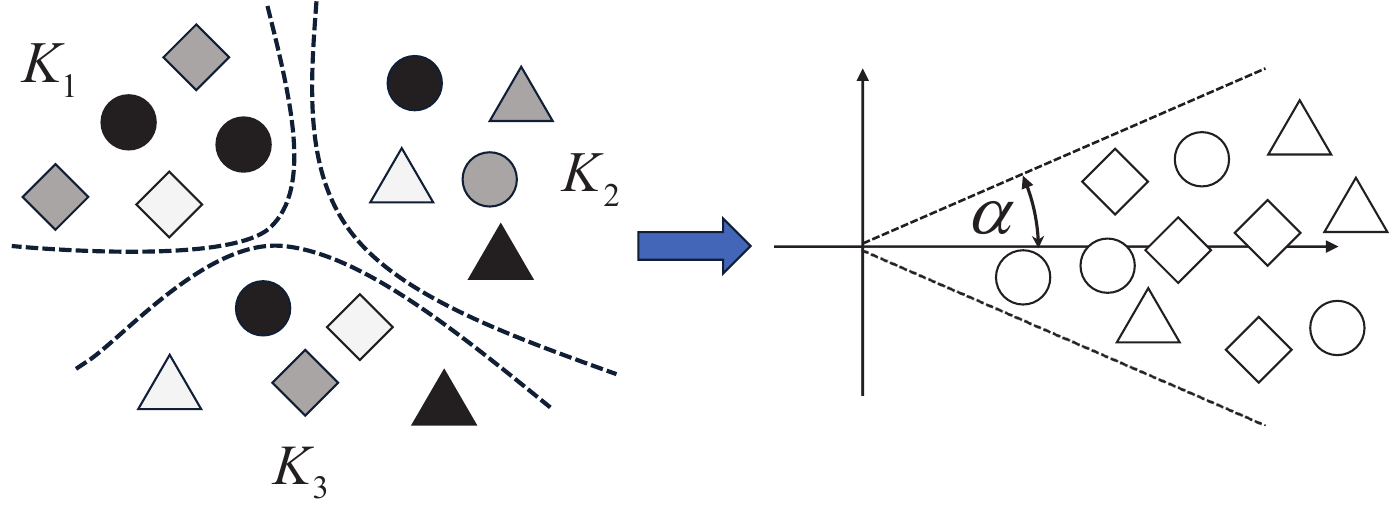}}
\caption{Matrices clustering via alignment properties, each shape represents an agent.}
\label{fig4}
\end{figure}

The mathematical formulation of the problem is presented as follows:

Given an angle $\alpha \in (-\pi/2,\pi/2)$ and a set of $k$ complex matrices $\mathcal{A} = \{A_1, A_2, \dots, A_k\}$, where $A_i \in \mathbb{C}^{m\times n}$ with $m \leq n$, we aim to partition $\mathcal{A}$ into $p$ disjoint clusters $\mathcal{C}=\{\mathcal{C}_1, \mathcal{C}_2, \dots, \mathcal{C}_{p}\}$, $\bigcup_{i=1}^{p} \mathcal{C}_i = \mathcal{A}$. The optimization problem is formulated as follows:
\begin{align*}
&\min_{\substack{p,\ \{\mathcal{C}_i\}}}
     \quad p  \label{eq:obj} \\[6pt]
&\text{s.t.}
    \quad\text{div}(\mathcal{C}_i)\leq \alpha,\quad i=1,\dots,p.
\end{align*}

The objective is to minimize the number of clusters $p$ while satisfying the above constraints.

Next, we demonstrate how to replace one set of clusters with another without increasing the cardinality of the cluster. This lemma will be employed in the following section to search for an optimal solution.

\begin{lemma}[Swapping Lemma]
\label{lem:swap-alpha-align}
Let $\mathcal{C}$ be a $\alpha$-alignable partition of matrix set $\mathcal{A}$ where $\mathcal{X}\subseteq\mathcal{C}$ is a partition of the subset $U\subseteq\mathcal{A}$ (note that $\bigcup_{\mathcal{X}_i\in\mathcal{X}}\mathcal{X}_i=U\subseteq\mathcal{A}$). Let $\mathcal{Y}$ be a set of disjoint clusters that cover $\bigcup_{\mathcal{Y}_i\in\mathcal{Y}}\mathcal{Y}_i=Q\subseteq\mathcal{A}$ such that $U\subseteq Q$ and $|\mathcal{Y}|\leq |\mathcal{X}|$.
Then
\begin{equation}  \mathcal{C}'=\bigl((\mathcal{C}\setminus\mathcal{X})\setminus (Q\setminus U)\bigr)\cup \mathcal{Y}
\end{equation}
is a new partition of $\mathcal{A}$ such that $|\mathcal{C}'|\leq|\mathcal{C}|$.
\end{lemma}
\begin{proof}
Define \(\mathcal{R} = \mathcal{C} \setminus \mathcal{X}\), which partitions \(\mathcal{A} \setminus U\). Since \(U \subseteq Q\), the union \(\mathcal{R} \cup \mathcal{Y}\) overlaps on \(Q \setminus U\). By the downward closedness property of \(\alpha\)-alignable clusters (Property \ref{downwardclosed}), removing the overlapping part \(Q \setminus U\) from \(\mathcal{R}\) preserves the structural validity of the remaining partition. Define \(\mathcal{R}' = \mathcal{R} \setminus (Q \setminus U)\), which partitions \(\mathcal{A} \setminus Q\).  

The new partition \(\mathcal{C}' = \mathcal{R}' \cup \mathcal{Y}\) covers \(\mathcal{A}\) because \(\mathcal{R}'\) covers \(\mathcal{A} \setminus Q\) and \(\mathcal{Y}\) covers \(Q\), with their union being disjoint. To bound the size, observe that \(|\mathcal{Y}| \leq |\mathcal{X}|\) by assumption, and \(|\mathcal{R}'| \leq |\mathcal{R}|\) since \(\mathcal{R}'\) is a subset of \(\mathcal{R}\). Therefore,  
\[
|\mathcal{C}'| = |\mathcal{R}'| + |\mathcal{Y}| \leq |\mathcal{R}| + |\mathcal{X}| = |\mathcal{C}|,
\]  
which completes the proof.
\end{proof}
\section{GLOBAL OPTIMA AND APPROXIMATION}
Before presenting our algorithm, we briefly address the problem’s complexity: this task is NP-hard. It can be reduced to the classical Minimum Clique Cover (MCC) problem (also NP-hard), with stricter constraints—global simultaneous alignment of matrices within each cluster is required, rather than merely pairwise alignment—making it at least as hard as its parent problem \cite{b16}. However, for small-scale instances (e.g., up to hundreds of elements), a recursive algorithm can still compute the minimum clustering within reasonable time frames \cite{b17}.

In the Minimum Clique Cover (MCC) problem, various optimization techniques—such as vertex reduction and modular decomposition—can be employed to accelerate algorithm convergence \cite{b18}. However, in the context of our problem, most of these conventional methods are inapplicable due to the unique features of matrix phases. Specifically, simplifying parts of the structure (e.g., removing vertices or edges) could disrupt the global alignment of matrices, rendering such optimizations ineffective \cite{b19}.
\subsection{Exact Recursive Minimization}
To address this challenge, we first introduce a weighted adjacency matrix $\mathbb{G}(S)$ that captures the pairwise phases similarity between matrices: given the angle $\alpha\in(-\pi/2,\pi/2)$, for any two matrices $A_i$ and $A_j$, $s_{i,j}=\frac{\pi}{2}-\text{div}(A_i,A_j)$ if they are $\alpha$-alignable; otherwise, $s_{i,j}=0$.

Here, $\text{div}(A_i,A_j)$ quantifies the similarity between $A_i$ and $A_j$, and $\alpha$-alignability ensures that only matrices with sufficient alignment potential are connected in the graph. By constructing $\mathbb{G}(S)$, we can precompute the structural properties of the agent network, which significantly facilitates the subsequent problem-solving process.
\begin{algorithm}[!h]
\caption{Branch-and-Recurse (BnR)} 
\label{BnR}
\KwIn{
\begin{tabular}[t]{r@{ : }l}
\(\mathbb{G}\)& The uncovered graph built by matrix S,\\
\(\mathcal{C}\)& Currently growing set of clusters,\\
\(\mathcal{C^*}\)& The smallest clustering found so far
\end{tabular}}
\KwOut{Minimum clustering \( \mathcal{C}^* \)}
\SetAlgoLined 
\DontPrintSemicolon

\For{$\mathbb{G}_i\in \text{Components}(\mathbb{G})$}{
    $\mathcal{C}^* \gets \bigcup_{\mathbb{G}_i} \text{BnR}(\mathbb{G}_i, \emptyset, \{\{v\}\,|\,v\in V(\mathbb{G}_i)\})$\;
    \Return $\mathcal{C}^*$\;
}

\If{$V(\mathbb{G})=\emptyset$}{
    \Return $\min(|\mathcal{C}|, |\mathcal{C}^*|)$\;
}

$v \gets $a matrix in $V(\mathbb{G})$ of minimum degree\;
\For{each maximal cluster $C_i \supseteq \{v\}$}{
    $\mathcal{C}^* \gets \text{BnR}(G \setminus C_i, \mathcal{C} \cup C_i, \mathcal{C}^*) $\;
}
\Return $\mathcal{C}^*$\;
\end{algorithm}

By incorporating the concept of $diversity$, we derive the Branch-and-Recurse (BnR) algorithm (Algorithm \ref{BnR}) to compute the minimum clustering for small-scale instances. The algorithm operates under the constraint that non-alignable matrices cannot coexist within the same cluster. Consequently, the graph is first partitioned into connected components, each of which is recursively solved independently, followed by merging the results.

The core intuition of BnR lies in leveraging the similarity matrix $\mathbb{G}(S)$ to guide the search. Specifically, at each step, the algorithm selects the vertex $v$ with the minimal degree among uncovered vertices as the root. A low degree implies weaker relationships (i.e., fewer alignable pairs) between $v$ and other matrices, making it easier to identify clusters by first separating weakly connected nodes. 

Having selected an appropriate root point $v$, the algorithm enumerates all maximal clusters $C$ containing $v$ from uncovered matrices, treating them as branches. 
The maximal cluster is defined similarly to the maximal clique: a maximal cluster represents the largest subset of matrices that can simultaneously $\alpha$-align; adding any additional matrix would violate this alignment property.
For each branch, the selected maximal cluster $C_i$ is removed from the uncovered subgraph $\mathbb{G}$, added to the current partial solution $\mathcal{C}$, and the process iterates recursively until all vertices are covered. When the recursive search exhausts all branches, the algorithm returns the minimum clustering found, as illustrated in Figure 5.
This branching is inspired by Gramm et al. \cite{b20}, who proposed branching on all maximal cliques containing an edge in the edge clique cover problem.

The optimal convergence of the algorithm can be derived from the Swapping Lemma \ref{lem:swap-alpha-align}. Suppose there exists a minimum clustering $\mathcal{C}^*$, for the matrix set $\mathcal{A}$, while an optimal branch $\mathcal{C}^*_{alg}$ is obtained through algorithm \ref{BnR}. Since the algorithm always selects the maximum cluster at each current root node, we can iteratively replace clusters in $\mathcal{C}^*$ with those from $\mathcal{C}^*_{alg}$ starting from the initial root, while guaranteeing non-increasing cluster cardinality. Ultimately, we can completely replace $\mathcal{C}^*$ with $\mathcal{C}^*_{alg}$ through this process, thereby proving the equivalence of these two solutions.
\begin{figure}[htbp]
\centerline{\includegraphics[width=0.45\textwidth,page=1]{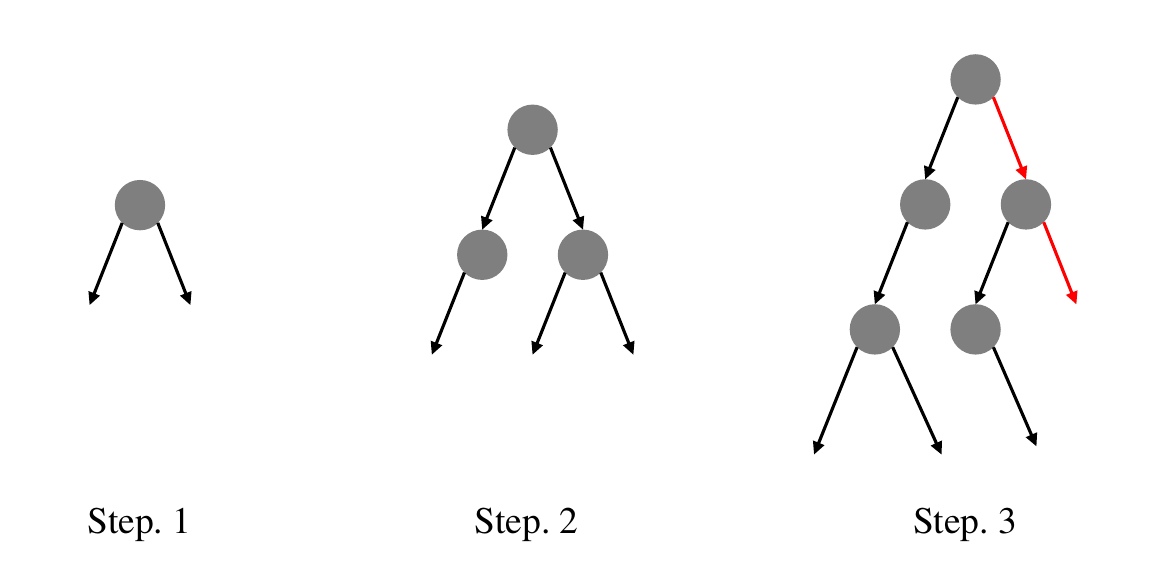}}
\caption{Steps of the BnR Algorithm: vertices represent root nodes, and arrows denote maximal clusters containing the corresponding root. The red path represents the optima at the end}
\label{fig5}
\end{figure}
\subsection{Randomized Approximation}
While the recursive algorithm guarantees exact optimal solutions, its exponential complexity ($O^*(c^n)$), where $c$ denotes the average branching factor per vertex) renders it impractical for large-scale instances. To address this, we integrate stochasticity into the search process, enabling efficient approximation of global optima within bounded computational budgets. By embedding a simulated annealing framework within the branch-and-prune architecture \cite{b22}, we develop a randomized variant (Algorithm \ref{HBnR}) that strategically balances computational cost and solution accuracy. This hybrid approach probabilistically prioritizes promising branches while discarding suboptimal paths, thereby accelerating convergence without sacrificing solution quality.
\begin{algorithm}[!h]
\caption{Heuristic Branch-and-Bound (HBnB)} 
\label{HBnR}
\KwIn{
\begin{tabular}[t]{r@{ : }l}
\(T,t\)& Global and branch temperatures,\\
\(\beta,\gamma\)& Cooling parameters of $T$ and $t$,\\
\(e\)& Termination temperature,\\
\(\mathbb{G}\)& The uncovered graph built by matrix S,\\
\(\mathcal{C}\)& Currently growing set of clusters,\\
\(\mathcal{C^*}\)& The smallest clustering found so far
\end{tabular}}
\KwOut{Minimum clustering \( \mathcal{C}^* \)}
\SetAlgoLined 
\DontPrintSemicolon
\For{$\mathbb{G}_i\in \text{Components}(\mathbb{G})$}{
    $\mathcal{C}^* \gets \bigcup_{\mathbb{G}_i} \text{HBnB}(\mathbb{G}_i, \emptyset, \{\{v\}\,|\,v\in V(\mathbb{G}_i)\})$\;
    \Return $\mathcal{C}^*$\;
}
\While{$T\geq e$}{
\If{$V(\mathbb{G})=\emptyset$}{
    $\mathcal{C}^*\gets\min(|\mathcal{C}|, |\mathcal{C}^*|)$\;
    $T\gets \beta T$\;
    $t\gets T$\;
    \Return $\text{Backtrack}(\mathcal{C},T)$\;
}
\If{$|\mathcal{C}|+\text{LowerBound}(\mathbb{G})\geq|\mathcal{C}^*|$}{
    $T\gets \beta T$\;
    $t\gets T$\;
    \Return $\text{Backtrack}(\mathcal{C},T)$\quad //\,Cut branch\;
}
$v \gets $a matrix in $V(\mathbb{G})$ of minimum degree\;
\For{each maximal cluster $C_i \supseteq \{v\}$}{
    $p_i\gets\text{Potential}(\mathbb{G},\mathcal{C},C_i)$
}
$C^*\gets\text{ChooseBranch}(p,t)$\;
$\mathbb{G}\gets\mathbb{G}\setminus C^*$\;
$\mathcal{C}\gets\mathcal{C}\cup C^*$\;
$t\gets\gamma t$
}
\end{algorithm}

\subsection*{Function Descriptions}
\begin{itemize}
    \item \text{Backtrack}($\mathcal{C}, \mathbb{G}, T$): 
    The algorithm randomly backtracks to a root node in the current partial solution path $\mathcal{C}$. Based on the position of that root node, it resets $\mathcal{C}$ and $\mathbb{G}$, and then continues the loop.
    The backtracking distance is regulated by the global temperature $T$. 
    A higher $T$ allows backtracking to more distant root nodes, while a lower $T$ restricts backtracking to recent root nodes.

    \item \text{LowerBound}($\mathbb{G}$): 
    Computes a lower bound for the minimum clustering by identifying the maximum independent set $\alpha(\mathbb{G})$ in the adjacency matrix $\mathbb{G}(S)$. 
    For a set of unclustered matrices, if there are $k$ matrices that cannot be aligned with one another, this implies that at least $k$ clusters are needed.
    Consequently, we obtain an important lower bound:
    \begin{equation}
        |\mathcal{C}_{\text{best}}| \geq \alpha(\mathbb{G}).
    \end{equation}

    This lower bound enables early pruning: if the sum of the lower bound $\alpha(\mathbb{G})$ and the number of clusters in the current partition $\mathcal{C}$ exceeds the current smallest clustering $\mathcal{C}^*$, then continuing the search along this branch cannot yield a better solution than $\mathcal{C}^*$. Consequently, this branch can be pruned.
    To balance quality and speed, we adopt the linear time algorithm by Chang et al. \cite{b21}, which efficiently finds a large independent set via a reducing-peeling strategy.

    \item \text{Potential}($\mathbb{G}, \mathcal{C}, C_i$): 
    Evaluates the branching potential for a candidate cluster $C_i$ using the formula:
    \begin{equation}
        \begin{split}                   p_i=\frac{\text{div}(C_i)}{\sum_{j=1}\text{div}(C_j)} +   
        \frac{1/\alpha(\mathbb{G}\setminus C_i)}{\sum_{j=1}1/\alpha(\mathbb{G}\setminus C_j)}.
    \end{split}
    \end{equation}
     This formula jointly considers the quality of the current branch and the quality of subsequent paths after branching.
    \item \text{ChooseBranch}($p, t$): 
    Converts potentials $p_i$ into branching probabilities via a temperature-dependent softmax formula:
    \begin{equation}
        P(C_i) = \frac{e^{p_i / t}}{\sum_j e^{p_j / t}}.
    \end{equation}
    Higher $t$ encourages uniform exploration, while lower $t$ greedily selects the highest-potential branch.
\end{itemize}

The key idea of Algorithm \ref{HBnR} is to maintain a partial solution path \(\mathcal{C}\) and iteratively backtrack and advance along it to search for better solutions, as illustrated in Fig. \ref{fig6}. At each branching point (root), the algorithm evaluates the potential of each candidate branch (i.e., maximal cluster \(C_i\)) and uses the current branch temperature \(t\) to stochastically select the next direction. This probabilistic selection balances computational cost and solution accuracy: higher \(t\) encourages diverse path sampling, while lower \(t\) prioritizes high-potential branches. 

Backtracking is triggered under two conditions: all data points are assigned to clusters, or further exploration cannot yield a solution better than the current best \(\mathcal{C}^*\). During backtracking, the algorithm probabilistically reverts to a root node in \(\mathcal{C}\), with the backtracking distance regulated by the global temperature \(T\). Higher \(T\) allows revisiting distant roots, promoting exploration of alternative paths, while lower \(T\) restricts backtracking to recent nodes, focusing on local refinements. To prevent cycling, all traversed paths are explicitly marked.

The dual-temperature mechanism ensures both randomness and convergence. The lower bound enables us to prune unpromising branches, thereby avoiding unnecessary exploration. As temperatures decrease, the algorithm becomes greedier: branch selection favors higher-potential clusters, and backtracking distances shorten, focusing on dense regions with unstable cluster structures (e.g., high-degree roots). This adaptive approach allows the algorithm to escape local optima while progressively refining solutions toward a high-quality local best. In addition, since the algorithm does not repeatedly explore paths that have already been searched, Algorithm \ref{HBnR} is equivalent to Algorithm \ref{BnR} when $T\to+\infty$.
\begin{figure}[htbp]
\centerline{\includegraphics[width=0.45\textwidth,page=2]{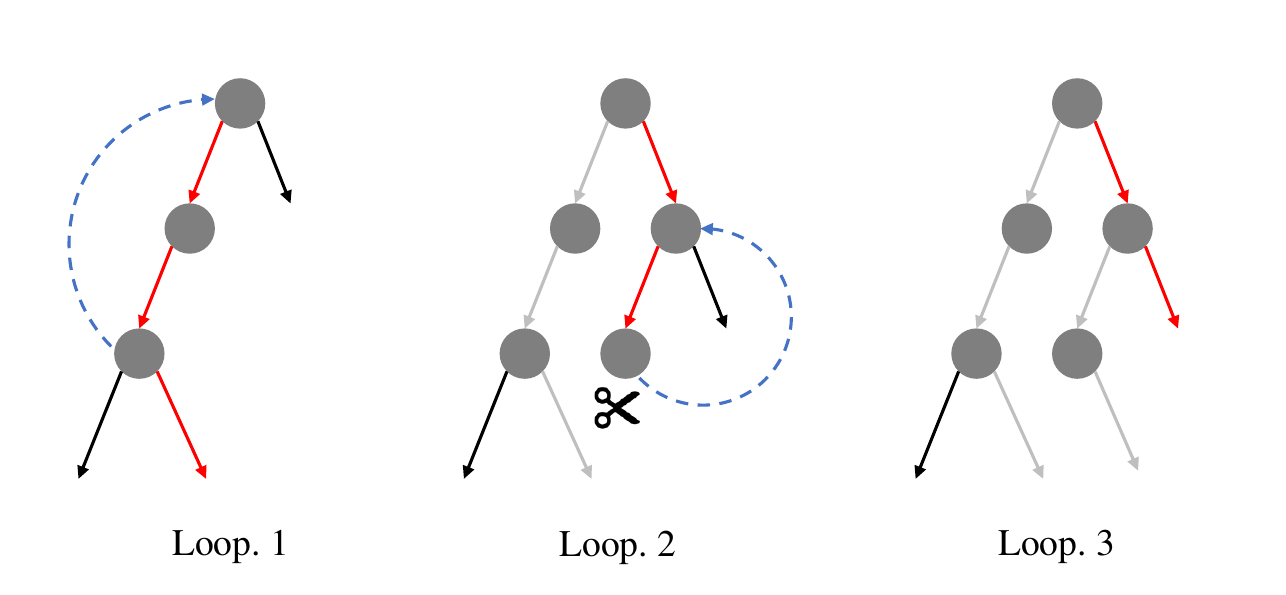}}
\caption{Steps of the HBnR Algorithm: the red, black, and gray paths represent the currently searched path $\mathcal{C}$, the incompletely searched path, and the fully searched path, respectively. The blue dashed line indicates the backtracking operation.}
\label{fig6}
\end{figure}
\section{SIMULATION}
\subsection{Synthetic Data Experiments}
To evaluate the convergence behavior of the proposed algorithms, we conducted extensive experiments on datasets of varying sizes. Due to the high computational complexity of Algorithm \ref{BnR} (BnR), we focus our analysis on Algorithm \ref{HBnR} (HBnB). The initial state of the clustering process is set to the worst-case scenario, where each matrix forms its own individual cluster. The test dataset consists of $n$ randomly generated three-dimensional complex matrices, ensuring a diverse and challenging evaluation environment.

In this simulation, we set \( T = 100 \), \( e = 10^{-5} \), $\beta=\gamma=0.9$.
The clustering results are illustrated in Figure \ref{fig7}, which demonstrates the evolution of the number of clusters as a function of iteration steps. It is evident that the number of clusters decreases significantly with increasing loop iterations, indicating the effectiveness of Algorithm 2 in progressively refining the clustering structure. 
\begin{figure}[htbp]
\centerline{\includegraphics[width=0.5\textwidth]{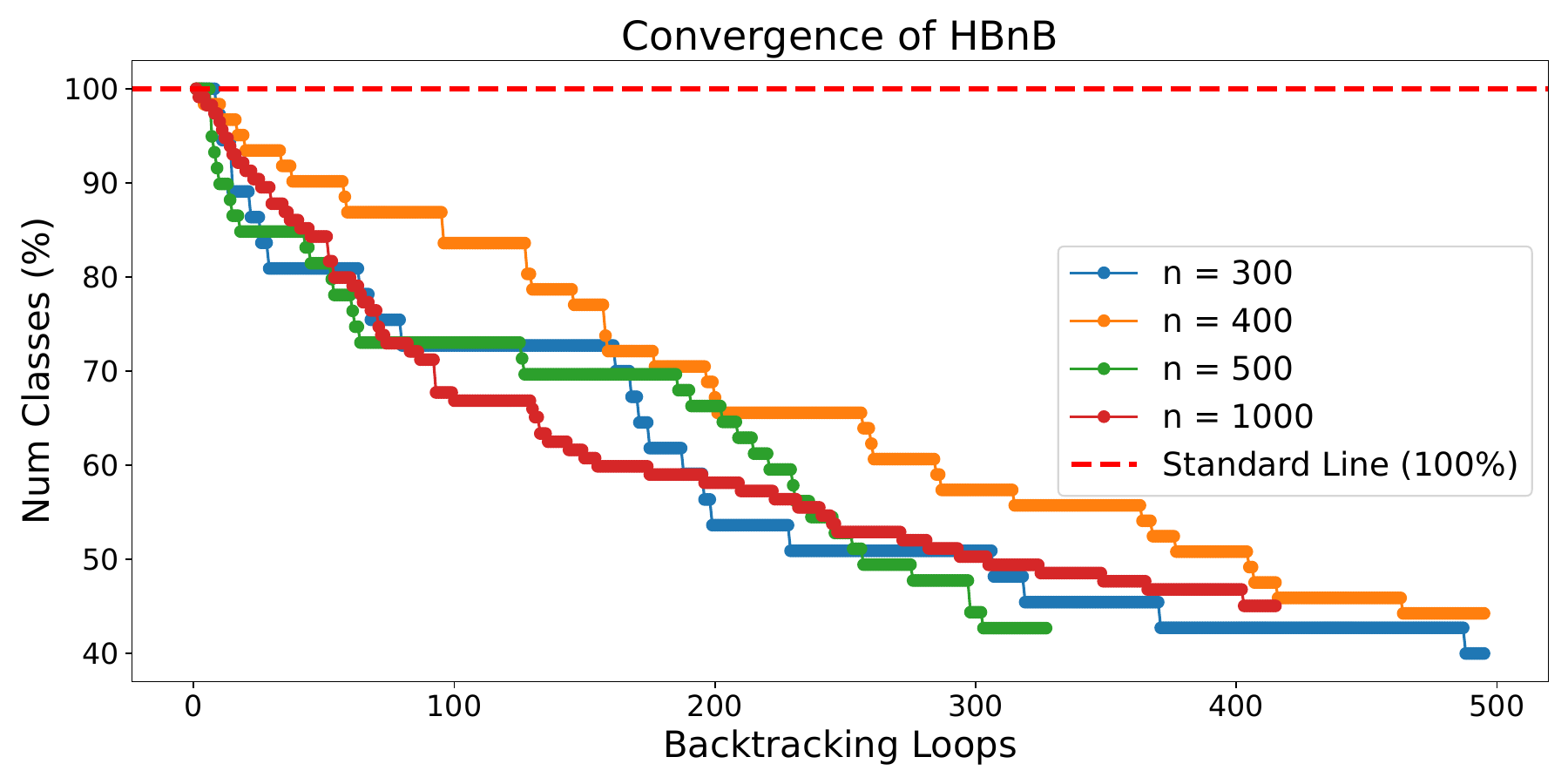}}
\caption{The convergence speed (percentage) of HBnB on data sets of different sizes.}
\label{fig7}
\end{figure}
\subsection{Real-World Scenario Analysis}
\begin{figure*}[htbp]
\centerline{\includegraphics[width=0.75\textwidth]{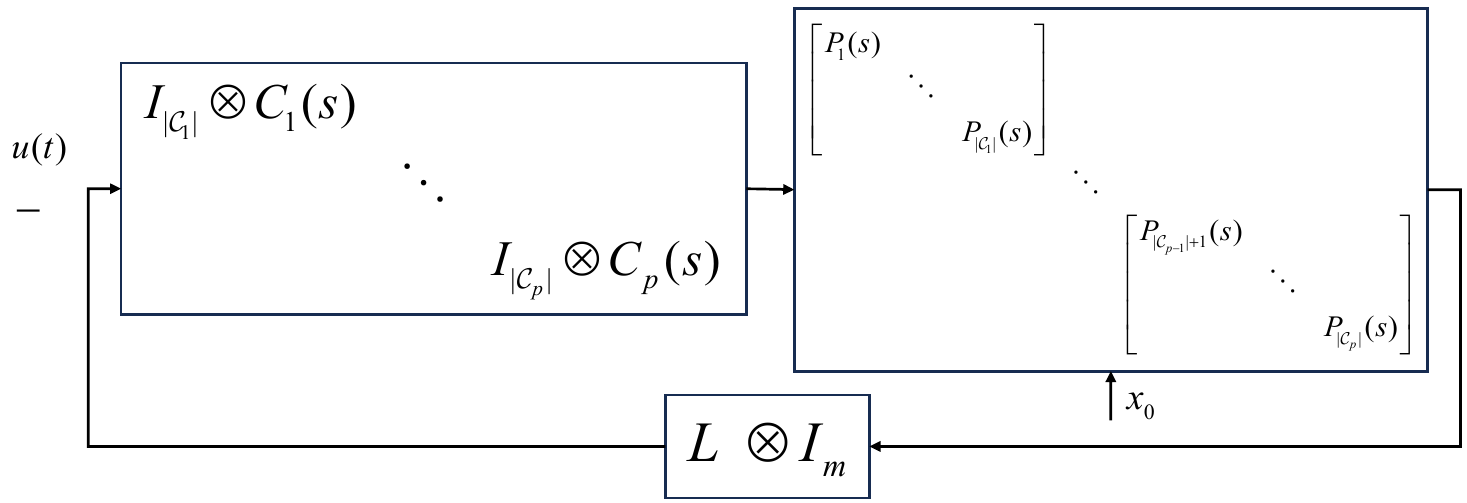}}
\caption{Synchronization with cluster-dependent controllers, transfer functions enclosed in the same block are considered members of one cluster, $\otimes$ denotes Kronecker product.}
\label{fig9}
\end{figure*}
\begin{figure*}[htbp]
    \centering
    \begin{subfigure}{0.49\textwidth}
        \centering
        \includegraphics[width=\linewidth]{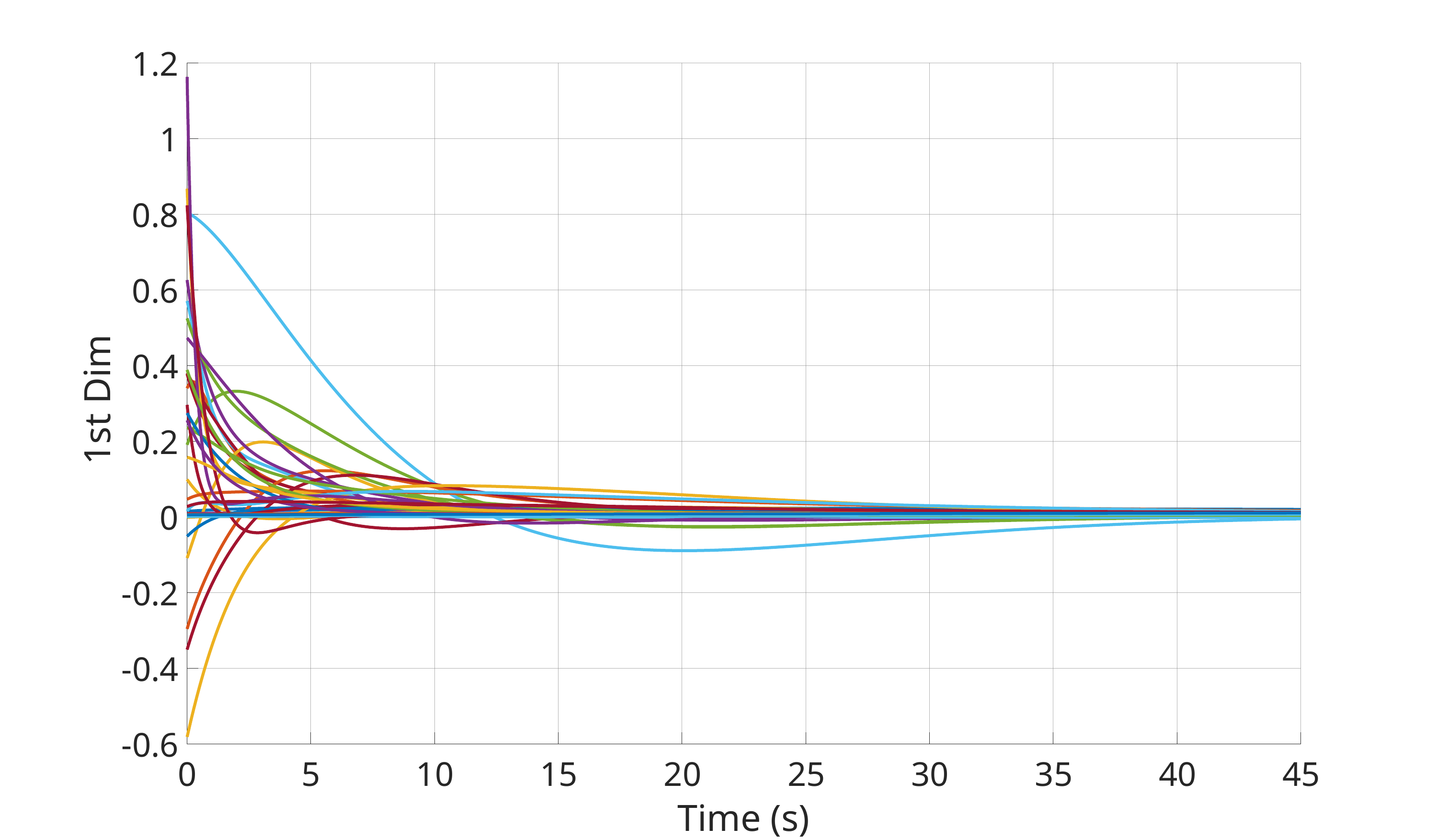}
        \caption{$y_1$}
        \label{fig10:a}
    \end{subfigure}
    \hfill
    \begin{subfigure}{0.49\textwidth}
        \centering
        \includegraphics[width=\linewidth]{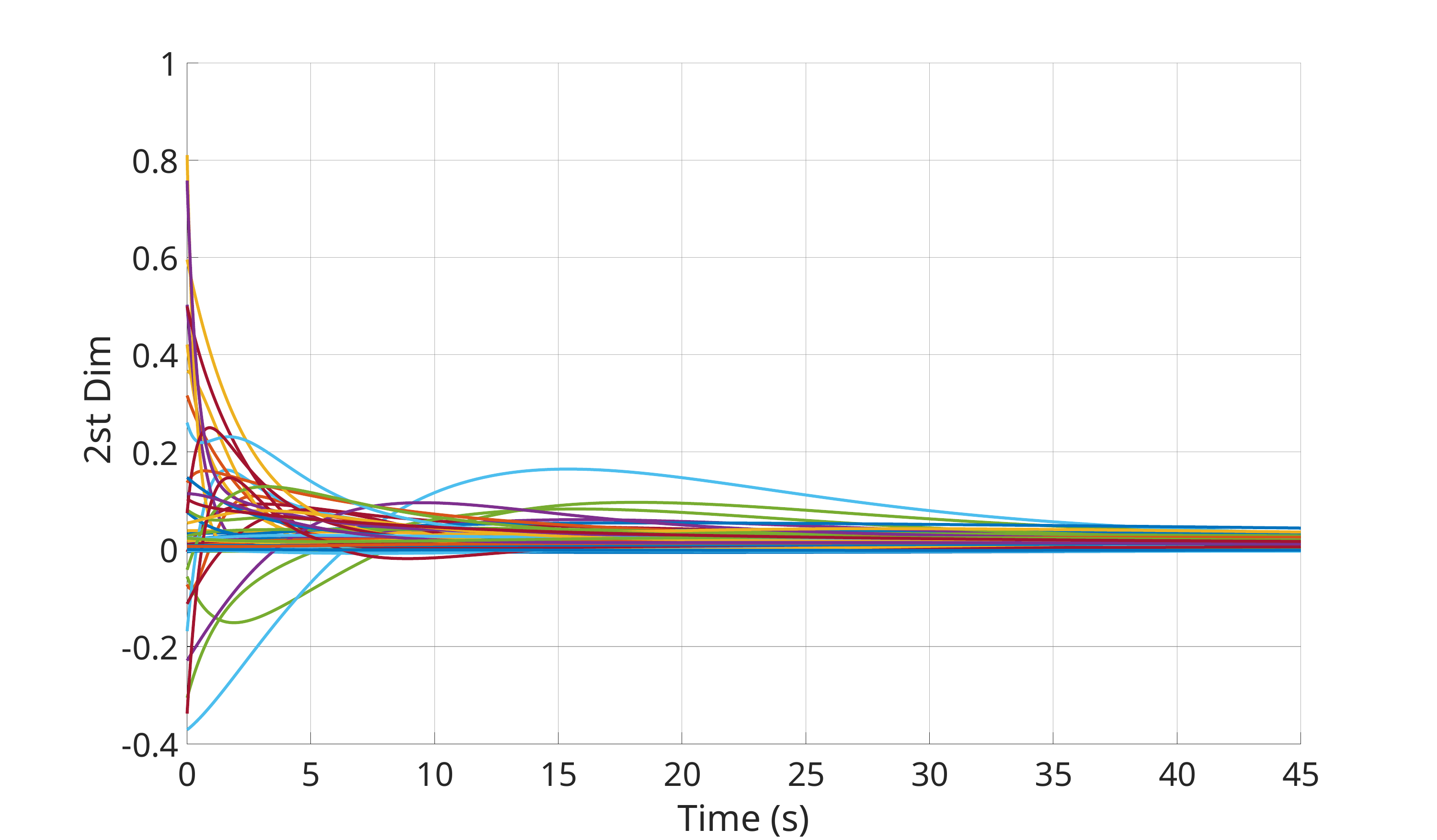}
        \caption{$y_2$}
        \label{fig10:b}
    \end{subfigure}
    \caption{Outputs of 50 agents.}
    \label{fig10}
\end{figure*}
To validate the practical efficacy of the proposed algorithm, we construct a 50-agent system where each agent operates as a dual-input-dual-output dynamic integrator:
\begin{equation}
    P(s)=\frac{M_i}{s},\quad i=1,\dots,50,
\end{equation}
where all $M_i\in\mathbb{C}^{2\times 2}$ are randomly generated.
The networked interactions are rigorously modeled through a randomly generated directed Laplacian matrix under a strongly connected topology. Let the essential phase of the Laplacian matrix $L$ of the network be denoted as $\phi_{ess}$.

According to Theorem of \cite{b10}, by partitioning these integrators into $p$ disjoint clusters $\{\mathcal{C}_i\}_{i=1}^p$, where the matrices within each cluster are simultaneously $\phi_{ess}$-alignable, we can solve for the corresponding alignment matrix $K_i$ for each cluster $\mathcal{C}_i$ using LMIs \ref{eq: LMIs}. Subsequently, $p$ synchronous controllers can be designed via the equation:
\begin{equation}
\label{design controller}
    C_i(s) = K_i.
\end{equation}

The connection between agents and controllers is illustrated in Fig.~\ref{fig9}. Here, \( \mathcal{C}_i \) denotes the \( i \)-th cluster after classification, and \( C_i(s) \) represents the controller corresponding to cluster \( \mathcal{C}_i \). In the block diagram on the right-hand side, each block along the diagonal represents the agents belonging to the same cluster, and \( L \) is the Laplacian matrix characterizing their interaction topology.  

In this simulation, we set \( T = 50 \), \( e = 10^{-5} \), $\beta=\gamma=0.9$, and initialized the solution to the worst-case scenario. The essential phase of the graph Laplacian matrix was computed as \( \phi_{\text{ess}}(\mathbb{G}(L)) = 0.9154 \) according to \eqref{essp}. After 30 minutes of execution, the algorithm partitioned 50 agents into 13 clusters, and corresponding controllers were designed using \eqref{design controller}. By connecting the agents and controllers, as shown in Fig. \ref{fig9}, synchronization of both outputs \( y_1 \) and \( y_2 \) was finally achieved, as demonstrated in Fig. \ref{fig10}. This confirms the effectiveness of the proposed control strategy.

\captionsetup[subfigure]{labelformat=empty}

\section{CONCLUSION}
This paper presents a novel phase-alignment-based framework for minimizing the number of controller types required to synchronize multi-agent systems. By leveraging the intrinsic phase properties of complex matrices, we introduce a clustering strategy that groups agents with aligned synchronization behaviors, thereby reducing redundant control inputs. We formalize the concept of matrix phase alignment and propose a minimum clustering problem constrained by phase similarity, enabling systematic controller reduction while preserving synchronization performance. 

Two algorithms are developed—Branch-and-Recurse (BnR) for exact global optimization in small-scale systems and Heuristic Branch-and-Bound (HBnB) for scalable approximations in large-scale scenarios. HBnB integrates simulated annealing to balance computational cost and solution accuracy, achieving near-optimal solutions with polynomial-time complexity. This work bridges the gap between theoretical phase analysis and practical control synthesis, offering a cost-effective solution for large-scale networks.


\end{document}